\documentclass[a4paper,11pt]{article}
\usepackage{geometry}
\usepackage{xspace}
\usepackage{url}
\usepackage[small]{caption}
\usepackage{graphicx}
\usepackage{amsmath}
\usepackage{amsthm}
\usepackage{booktabs}
\usepackage[numbers]{natbib}

\usepackage[linesnumbered,ruled,vlined]{algorithm2e}
\RequirePackage{comment}
\usepackage{mathtools}
\RequirePackage{comment}
\usepackage{xifthen}
\usepackage{arydshln}
\usepackage[hidelinks]{hyperref}
\usepackage[capitalise]{cleveref} 
\usepackage{xcolor}
\usepackage{amssymb}
\usepackage{dsfont}
\usepackage{colortbl}
\usepackage{graphics}
\usepackage{float}
\usepackage{color}
\usepackage{xcolor}
\usepackage{array}
\usepackage{multirow}
\usepackage[shortlabels]{enumitem}
\usepackage{amsmath,latexsym,color,amsthm,authblk}

\newtheorem{proposition}{\bf Proposition}

\newtheorem{theorem}{\bf Theorem}

\newtheorem{definition}{\bf Definition}

\crefname{theorem}{theorem}{\bf Theorem}
\crefname{example}{example}{\bf Example}
\crefname{observation}{observation}{\bf Observation}
\crefname{lemma}{lemma}{\bf Lemma}
\crefname{corollary}{corollary}{\bf Corollary}
\crefname{proposition}{proposition}{\bf Proposition}
\crefname{definition}{definition}{\bf Definition}
\crefname{claim}{claim}{\bf Claim}
\crefname{reductionrule}{reduction rule}{\bf Reduction rule}

\newcommand{\instance}{\ensuremath{I}\xspace}
\newcommand{\suml}{\sum\limits_}

\DeclarePairedDelimiter\floor{\lfloor}{\rfloor}

\newcommand{\curly}[1]{\ensuremath{\{#1\}}\xspace}
\newcommand{\NPH}{\ensuremath{\mathsf{NP}}-hard\xspace}

\newcommand{\yes}{{\sc Yes}\xspace}

\newcommand{\FPT}{\ensuremath{\mathsf{FPT}}\xspace}

\newcommand{\bud}{\ensuremath{b}\xspace}
\newcommand{\proj}{\ensuremath{P}\xspace}
\newcommand{\degreeset}{\ensuremath{\operatorname{D}}\xspace}
\newcommand{\degreeproj}{\ensuremath{\mathcal{D}}\xspace}
\newcommand{\costfunction}{\ensuremath{\operatorname{c}}\xspace}
\newcommand{\voters}{\ensuremath{N}\xspace}

\newcommand{\valid}{\ensuremath{\mathcal{V}}\xspace}

\newcommand{\lbounds}{\ensuremath{\mathcal{L}}\xspace}
\newcommand{\eachlb}{\ensuremath{l}\xspace}
\newcommand{\ubounds}{\ensuremath{\mathcal{H}}\xspace}
\newcommand{\eachub}{\ensuremath{h}\xspace}
\newcommand{\fullinstance}{\ensuremath{\langle \voters,\degreeproj,\costfunction,\bud,\lbounds,\ubounds\rangle}\xspace}
\newcommand{\ii}{\ensuremath{i}\xspace}

\newcommand{\pa}{\ensuremath{j}\xspace}
\newcommand{\pdof}[2]{%
  \ifthenelse{\isempty{#1}}%
    {\ifthenelse{\isempty{#2}}
    {\ensuremath{P_{j}^{t}}\xspace}
    {\ensuremath{P_{j}^{#2}}\xspace}
    }
    {\ifthenelse{\isempty{#2}}{\ensuremath{P_{#1}^{t}}\xspace} 
    {\ensuremath{P_{#1}^{#2}}\xspace} 
    }
}
\newcommand{\cof}[2]{%
  \ifthenelse{\isempty{#1}}%
    {\ifthenelse{\isempty{#2}}
    {\ensuremath{c_{j}^{t}}\xspace}
    {\ensuremath{c_{j}^{#2}}\xspace}
    }
    {\ifthenelse{\isempty{#2}}{\ensuremath{c_{#1}^{t}}\xspace} 
    {\ensuremath{c_{#1}^{#2}}\xspace} 
    }
}
\newcommand{\conrange}{\ensuremath{\tau_\pa}\xspace}
\newcommand{\maxconrange}{\ensuremath{\overline{\conrange}}\xspace}
\newcommand{\minconrange}{\ensuremath{\underline{\conrange}}\xspace}
\newcommand{\qof}[2]{%
  \ifthenelse{\isempty{#1}}%
    {\ifthenelse{\isempty{#2}}
    {\ensuremath{q_{j}^{t}}\xspace}
    {\ensuremath{q_{j}^{#2}}\xspace}
    }
    {\ifthenelse{\isempty{#2}}{\ensuremath{q_{#1}^{t}}\xspace} 
    {\ensuremath{q_{#1}^{#2}}\xspace} 
    }
}
\newcommand{\qmax}{\ensuremath{q_m}\xspace}
\newcommand{\qsum}{\ensuremath{q_\sigma}\xspace}
\newcommand{\csetof}[1]{
\ifthenelse{\isempty{#1}}
{\ensuremath{\costfunction\!\left(S\right)}\xspace}
{\ensuremath{\costfunction\!\left(#1\right)}\xspace}
}
\newcommand{\mdof}[1]{
\ifthenelse{\isempty{#1}}
{\ensuremath{t_{j}}\xspace}
{\ensuremath{t_{#1}}\xspace}
}
\newcommand{\dsetof}[1]{
\ifthenelse{\isempty{#1}}
{\ensuremath{\degreeset\!\left(P_{j}\right)}\xspace}
{\ensuremath{\degreeset\!\left(P_{#1}\right)}\xspace}
}
\newcommand{\utof}[2]{
\ifthenelse{\isempty{#1}}
{
\ifthenelse{\isempty{#2}}
{\ensuremath{u_i\left(S\right)}\xspace}
{\ensuremath{u_i\left(S,#2\right)}\xspace}
}
{
\ifthenelse{\isempty{#2}}
{\ensuremath{u_{#1}\left(S\right)}\xspace}
{\ensuremath{u_{#1}\left(S,#2\right)}\xspace}
}
}
\newcommand{\utofd}[2]{
\ifthenelse{\isempty{#1}}
{
\ifthenelse{\isempty{#2}}
{\ensuremath{u_i\left(S'\right)}\xspace}
{\ensuremath{u_i\left(S',#2\right)}\xspace}
}
{
\ifthenelse{\isempty{#2}}
{\ensuremath{u_{#1}\left(S'\right)}\xspace}
{\ensuremath{u_{#1}\left(S',#2\right)}\xspace}
}
}
\newcommand{\dutof}[2]{
\ifthenelse{\isempty{#1}}
{
\ifthenelse{\isempty{#2}}
{\ensuremath{d_i\left(S\right)}\xspace}
{\ensuremath{d_i\left(S,#2\right)}\xspace}
}
{
\ifthenelse{\isempty{#2}}
{\ensuremath{d_{#1}\left(S\right)}\xspace}
{\ensuremath{d_{#1}\left(S,#2\right)}\xspace}
}
}
\newcommand{\dutofd}[2]{
\ifthenelse{\isempty{#1}}
{
\ifthenelse{\isempty{#2}}
{\ensuremath{d_i\left(S'\right)}\xspace}
{\ensuremath{d_i\left(S',#2\right)}\xspace}
}
{
\ifthenelse{\isempty{#2}}
{\ensuremath{d_{#1}\left(S'\right)}\xspace}
{\ensuremath{d_{#1}\left(S',#2\right)}\xspace}
}
}
\newcommand{\variance}{variance coefficient\xspace}
\newcommand{\vpara}{\ensuremath{\gamma}\xspace}
\newcommand{\scorefunction}{\ensuremath{\operatorname{s}\!}\xspace}
\newcommand{\scoreof}[2]{
\ifthenelse{\isempty{#1}}
{
\ifthenelse{\isempty{#2}}
{\ensuremath{\scorefunction\left(\pdof{}{}\right)}\xspace}
{\ensuremath{\scorefunction\left(\pdof{}{#2}\right)}\xspace}
}
{
\ifthenelse{\isempty{#2}}
{\ensuremath{\scorefunction\left(\pdof{#1}{}\right)}\xspace}
{\ensuremath{\scorefunction\left(\pdof{#1}{#2}\right)}\xspace}
}
}
\newcommand{\msof}[1]{
\ifthenelse{\isempty{#1}}
{\ensuremath{\alpha\!\left(\pdof{}{}\right)}\xspace}
{\ensuremath{\alpha\!\left(\pdof{}{#1}\right)}\xspace}
}
\newcommand{\cardof}[1]{
\ifthenelse{\isempty{#1}}
{\ensuremath{|S|}\xspace}
{\ensuremath{|#1|}\xspace}
}
\newcommand{\lowerb}[2]{
\ifthenelse{\isempty{#1}}
{
\ifthenelse{\isempty{#2}}
{\ensuremath{\eachlb_i(j)}\xspace}
{\ensuremath{\eachlb_i(#2)}\xspace}
}
{
\ifthenelse{\isempty{#2}}
{\ensuremath{\eachlb_{#1}(j)}\xspace}
{\ensuremath{\eachlb_{#1}(#2)}\xspace}
}
}
\newcommand{\upperb}[2]{
\ifthenelse{\isempty{#1}}
{
\ifthenelse{\isempty{#2}}
{\ensuremath{\eachub_i(j)}\xspace}
{\ensuremath{\eachub_i(#2)}\xspace}
}
{
\ifthenelse{\isempty{#2}}
{\ensuremath{\eachub_{#1}(j)}\xspace}
{\ensuremath{\eachub_{#1}(#2)}\xspace}
}
}
\newcommand{\appof}[2]{
\ifthenelse{\isempty{#1}}
{
\ifthenelse{\isempty{#2}}
{\ensuremath{A_i(j)}\xspace}
{\ensuremath{A_i(#2)}\xspace}
}
{
\ifthenelse{\isempty{#2}}
{\ensuremath{A_{#1}(j)}\xspace}
{\ensuremath{A_{#1}(#2)}\xspace}
}
}
\newcommand{\chodof}[2]{
\ifthenelse{\isempty{#1}}
{
\ifthenelse{\isempty{#2}}
{\ensuremath{S(j)}\xspace}
{\ensuremath{#2(j)}\xspace}
}
{
\ifthenelse{\isempty{#2}}
{\ensuremath{S(#1)}\xspace}
{\ensuremath{#2(#1)}\xspace}
}
}
\newcommand{\chocof}[2]{
\ifthenelse{\isempty{#1}}
{
\ifthenelse{\isempty{#2}}
{\ensuremath{c^S(j)}\xspace}
{\ensuremath{c^{#2}(j)}\xspace}
}
{
\ifthenelse{\isempty{#2}}
{\ensuremath{c^S(#1)}\xspace}
{\ensuremath{c^{#2}(#1)}\xspace}
}
}
\newcommand{\pbrule}{\ensuremath{R}\xspace}
\newcommand{\nrule}{\ensuremath{R_{|S|}}\xspace}
\newcommand{\tstar}{\ensuremath{t^*}\xspace}
\newcommand{\crule}{\ensuremath{R_{\costfunction(S)}}\xspace}
\newcommand{\slimit}{\ensuremath{\delta}\xspace}
\newcommand{\ccaprule}{\ensuremath{R_{\widehat{\costfunction(S)}}}\xspace}
\newcommand{\drule}{\ensuremath{R_{\|\;\|}}\xspace}
\newcommand{\winnersof}[2]{%
  \ifthenelse{\isempty{#2}}%
    {\ensuremath{\pbrule\left(\instance\right)}\xspace}
    {\ensuremath{\pbrule\left(\instance#2\right)}\xspace}
}
\newcommand{\winners}[2]{%
  \ifthenelse{\isempty{#2}}%
    {\ensuremath{\operatorname{\winnerfunction\!}\left(#1,\instance\right)}\xspace}
    {\ensuremath{\operatorname{\winnerfunction\!}\left(#1,\instance#2\right)}\xspace}
}

\makeatletter
\newcommand\mathcircled[1]{%
  \mathpalette\@mathcircled{#1}%
}
\newcommand\@mathcircled[2]{%
  \tikz[baseline=(math.base)] \node[draw,circle,inner sep=1pt] (math) {$\m@th#1#2$};%
}
\makeatother
\makeatletter
\newcommand{\thickhline}{%
    \noalign {\ifnum 0=`}\fi \hrule height 1.3pt
    \futurelet \reserved@a \@xhline
}
\newcommand{\toothickhline}{%
    \noalign {\ifnum 0=`}\fi \hrule height 2pt
    \futurelet \reserved@a \@xhline
}
\newcolumntype{"}{@{\hskip\tabcolsep\vrule width 1.3pt\hskip\tabcolsep}}
\newcolumntype{x}{@{\hskip\tabcolsep\vrule width 2pt\hskip\tabcolsep}}
\makeatother

\makeatletter
\newcommand*{\rom}[1]{\expandafter\@slowromancap\romannumeral #1@}
\makeatother
\setlist[enumerate]{nosep}

\title{Participatory Budgeting With Multiple Degrees of Projects And Ranged Approval Votes}

\author{Gogulapati Sreedurga}
\date{Indian Institute of Science}
\begin{document}

\maketitle

\begin{abstract}
    In an indivisible participatory budgeting (PB) framework, we have a limited budget that is to be distributed among a set of projects, by aggregating the preferences of voters for the projects. All the prior work on indivisible PB assumes that each project has only one possible cost. In this work, we let each project have a set of permissible costs, each reflecting a possible degree of sophistication of the project. Each voter approves a range of costs for each project, by giving an upper and lower bound on the cost that she thinks the project deserves. The outcome of a PB rule selects a subset of projects and also specifies their corresponding costs. We study different utility notions and prove that the existing positive results when every project has exactly one permissible cost can also be extended to our framework where a project has several permissible costs. We also analyze the fixed parameter tractability of the problem. Finally, we propose some important and intuitive axioms and analyze their satisfiability by different PB rules. We conclude by making some crucial remarks.
\end{abstract}
\section{Introduction}\label{sec: intro}
Participatory budgeting (PB) is a democratic voting paradigm that aggregates the opinions of citizens while deciding on how to fund the public projects \cite{cabannes2004participatory,shah2007participatory,sintomer2008participatory,wampler2010participatory,rocke2014framing}. 

Divisible PB assumes that the costs of the projects are totally flexible and any amount can be allocated to each of them. The existing work on indivisible PB, on the other hand, assumes that every project has a fixed cost that is to be allocated to the project if it is selected.  
However, many times in real-world, each project can be implemented upto different levels of sophistication. For example, a building can be built with wood or stone, depending on the amount allocated to it. That is, each project has a set of possible and permissible costs. Each cost in this set corresponds to a degree of sophistication of the project. In addition to choosing the set of projects that are to be funded, the mechanism designer must also choose a permissible cost for each of the selected projects. This idea of having multiple permissible costs for each project was first proposed in a survey by Aziz and Shah \cite{aziz2021participatory}, in which the authors called each of these costs a \emph{degree of completion}. However, the idea still remained to be studied and our paper bridges this gap.

Preference elicitation methods typically studied in any voting framework include approval votes, ordinal votes, and cardinal votes. These methods also continue to be the most studied preference elicitation approaches in PB \cite{shapiro2017participatory,aziz2018proportionally,talmon2019framework,rey2020designing,jain2020participatory,benade2021preference,aziz2021proportionally,fairstein2022welfare,sreedurga2022maxmin}. However, PB is a setting with several attributes like costs attached to the projects. The preferences and utilities of the voters are thus much more complex. This propels the need to devise preference models specific to PB, as pointed out by Aziz and Shah \cite{aziz2021participatory}. One step in this direction is the introduction of a special case of approval votes, called knapsack votes, by Goel et al. \cite{goel2019knapsack}, where each voter reports her most favorite budget division. This idea however is criticised for its assumption that any project not in this division yields no utility to the voter.

Our paper is another step in the direction of preference modeling for PB. We introduce the approach of \emph{ranged approval votes}, which strictly generalizes the approval votes. Each voter reports a lower bound and an upper bound on the cost that she thinks each project deserves. All the bounds are initially set to $0$ by default. Voting proceeds in two steps. In the first step, the voter starts by approving the projects she likes. For these approved projects, only the upper bounds will automatically change to the highest permissible cost. In the second step, the voter may \emph{optionally} change bounds for some of these approved projects, if she wishes to have a say on the amount they deserve. Note that this is cognitively not much more demanding than the approval votes since we do not force the voters to report the bounds. 
Notably, all our computational results can also be extended to more general utility functions with minor tweaks. Nevertheless, we present the results for ranged approval votes due to their cognitive simplicity and natural relevance in the model.

It needs to be mentioned that a work by Goel et al. \cite{goel2019knapsack} views divisible PB as a model of indivisible PB where every unit of money in the cost corresponds to one possible degree of sophistication, and proposes a greedy algorithm. Their work however imposes knapsack constraint on each vote as well as assumes all lower bounds to be zero. Their model hence forms a very restricted special case of ours.

\paragraph{\textbf{Our contributions}} The primary goal of this paper is to systematically study the PB model where each project has a set of permissible costs. Such a study has been conducted by Talmon and Faliszewski \cite{talmon2019framework} for approval-based PB, which is a special case of our model with every project having only one permissible cost. We generalize the PB rules and all the positive results in \cite{talmon2019framework} to our model. Namely, we propose polynomial-time, pseudo-polynomial time, and PTAS algorithms. Followed by this, we present some results on parameter tractability using a recently introduced parameter scalable limit \cite{sreedurga2022maxmin} and another novel parameter we introduce in the further sections. It needs to be highlighted that, as a part of our study, we also introduce and investigate some novel utility notions that are specific to our model and are not extensions of the literature. Finally, we propose some budgeting axioms for our model and examine their satisfiability with respect to our PB rules. All the axiomatic results are summarized in \Cref{tab: results}. We conclude by discussing the impact of our computational and axiomatic results and pointing out future directions.
\begin{table}[]
\begin{tabular}{l"c|c|c|c}
\textbf{PROPERTIES\textbackslash{}RULES} & $\boldsymbol{\nrule}$     & $\boldsymbol{\crule}$ & $\boldsymbol{\ccaprule}$ & $\boldsymbol{\drule}$ \\
\thickhline
\textbf{Degree-efficient}                & $\times$     & $\times$                                 & $\checkmark$                                & $\times$                                 \\
\hline
\textbf{Shrink-resitant}                 & $\checkmark$ & $\checkmark$                             & $\checkmark$                                & $\checkmark$                             \\
\hline
\textbf{Discount-proof}                  & $\checkmark$ & $\times$                                 & $\times$                                    & $\times$                                 \\
\hline
\textbf{Range-abiding}                   & $\checkmark$ & $\times$                                 & $\times$                                    & $\checkmark$                             \\
\hline
\textbf{Range-converging}                & $\checkmark$ & $\checkmark$                             & $\checkmark$                                & $\checkmark$                             \\
\hline
\textbf{Range-unanimous}                 & $\checkmark$ & $\times$                                 & $\times$                                    & $\checkmark$                             \\
\hline
\textbf{Upper bound-sensitive}           & $\checkmark$ & $\checkmark$                             & $\times$                                    & $\checkmark$                             \\
\hline
\textbf{Lower bound-sensitive}           & $\times$     & $\times$                                 & $\times$                                    & $\checkmark$               
\end{tabular}
\caption{Results for budgeting axioms in \Cref{sec: axioms}}
\label{tab: results}
\end{table}
\paragraph{\textbf{Organization of the paper}} We start by introducing the mathematical model in \Cref{sec: model}. We define different utility notions (some extended from the existing ones in the literature and some new) and the corresponding PB rules in \Cref{sec: utilities}. In \Cref{sec: computational}, we analyze the computational complexity of our PB rules and suggest ways to cope up with intractabilities. In \Cref{sec: axioms}, we define budgeting axioms for our model and examine their satisfiability. 

\section{Model}\label{sec: model}
In our model, a budget \bud is given. We denote the set of $n$ voters by $\voters = \curly{1,\ldots,n}$ and the set of $m$ projects by \proj. Each project $\pa \in \proj$ has \mdof{} possible degrees of sophistication captured by the set $\dsetof{} = \curly{\pdof{}{0},\pdof{}{1}, \ldots,\pdof{}{\mdof{}}}$. The cost of each degree \pdof{}{} is indicated by \cof{}{}. We assume that \cof{}{0} is zero for all $\pa \in \proj$ and it corresponds to not funding the project \pa.

Each voter $i \in \voters$ reports for every project \pa, a lower bound \lowerb{}{} and an upper bound \upperb{}{} such that $\lowerb{}{},\upperb{}{} \in \curly{\cof{}{0},\ldots,\cof{}{\mdof{}}}$ and $\lowerb{}{} \leq \upperb{}{}$. Let \lbounds and \ubounds respectively denote the collection all the lower bounds and upper bounds reported by all the voters. 

Let \degreeproj denote the set of all the possible degrees of all projects, or in other words, $\degreeproj = \bigcup_{\pa \in \proj}{\dsetof{}}$. We denote the cost of a set $S \subseteq \degreeproj$, $\sum_{\pdof{}{} \in S}{\cof{}{}}$, by \csetof{}. Given a subset $S \subseteq \degreeproj$, we use \chodof{}{} to denote the chosen degree(s) of project \pa in $S$. In other words, $\chodof{}{} = S \cap \dsetof{}$. We use the shorthand notation \chocof{}{} to denote \csetof{\chodof{}{}}. We say a subset $S \subseteq \degreeproj$ is \emph{\textbf{valid}} if $\csetof{} \leq \bud$ and $\cardof{\chodof{}{}} = 1$ for every $\pa \in \proj$. Let \valid denote the collection of all the valid subsets.

A PB instance with multiple degrees of sophistication is denoted by $\instance = \fullinstance$. Given an instance, the objective of a PB rule $\pbrule$ is to output a valid subset $S \in \valid$. 

\subsection{The PB Rules}\label{sec: utilities}
The PB rules we study are utilitarian rules. This implies that, given a function $u$ that measures utility a voter derives from a subset of projects, the corresponding PB rule outputs a valid set of projects that maximizes the sum of utilities of all the voters, i.e., it outputs some $S \in \valid$ such that $\sum_{i \in \voters}{\utof{}{}}$ is maximized.
We say a subset $S$ of projects is selected under a PB rule \pbrule if it maximizes the total utility of the voters. We say a project $\pdof{}{} \in \degreeproj$ \textbf{\emph{wins}} under a PB rule \pbrule if it belongs to some set that can be selected under \pbrule. Let \winnersof{\pbrule}{} be the collection of all the  projects that win under the PB rule \pbrule, given an instance \instance.

Now, we only need to define different utility functions. Given the lower and upper bounds reported by the voters, we define the utility of a voter from a valid set $S \in \valid$ in four ways. For each of these utility functions, we also give a shorthand notation for the utilitarian PB rule associated to it.
\begin{enumerate}
    \item \textbf{Cardinal utility ($\boldsymbol{\nrule}$ rule)}: Each voter \ii derives a utility of $1$ from a project \pa if the cost of the chosen degree falls within the bounds specified by the voter. Thus, $\utof{}{} = \cardof{\pa \in \proj : \lowerb{}{} \leq \chocof{}{} \leq \upperb{}{}}, \chocof{}{} \neq 0$.
    \item \textbf{Cost utility ($\boldsymbol{\crule}$ rule)}: A voter \ii derives a utility of \chocof{}{} from a project \pa if its value falls within the bounds specified by the voter. \begin{center}
        $\utof{}{} = \suml{\pa: \lowerb{}{} \leq \chocof{}{} \leq \upperb{}{}}{\chocof{}{}}$.
    \end{center}
    \item \textbf{Cost capped utility ($\boldsymbol{\ccaprule}$ rule)}: Each voter \ii derives a utility of \chocof{}{} from a project \pa if the value falls within the bounds reported by her, a utility of \upperb{}{} if $\chocof{}{} > \upperb{}{}$, and a utility of $0$ if $\chocof{}{} < \lowerb{}{}$. That is, with slight abuse of notation, $\utof{}{} = \sum_{\pa \in \proj}{\utof{}{j}}$, where \utof{}{j} is defined as:
    \begin{align*}
        \utof{}{j} \!=\!
        \begin{cases} 
        0 &  \chocof{}{} < \lowerb{}{}\\
        \chocof{}{} & \lowerb{}{} \leq \chocof{}{} \leq \upperb{}{}\\
        \upperb{}{} & \text{otherwise}
        \end{cases}
    \end{align*}
    \item \textbf{Distance disutility ($\boldsymbol{\drule}$ rule)}: From every project $\pa$, each voter derives a disutility of $0$ if the value falls within the bounds reported by her, a disutility of $\chocof{}{}-\upperb{}{}$ if $\chocof{}{} > \upperb{}{}$, and a disutility of $\lowerb{}{}-\chocof{}{}$ if $\chocof{}{} < \lowerb{}{}$. That is, with slight abuse of notation, $\dutof{}{} = \sum_{\pa \in \proj}{\dutof{}{j}}$, where \dutof{}{j} is defined as follows:
    \begin{align*}
        \dutof{}{j} \!=\!
        \begin{cases} 
        \lowerb{}{}-\chocof{}{} &  \chocof{}{} < \lowerb{}{}\\
        0 & \lowerb{}{} \leq \chocof{}{} \leq \upperb{}{}\\
        \chocof{}{}-\upperb{}{} & \text{otherwise}
        \end{cases}
    \end{align*}
    The corresponding PB rule $\boldsymbol{\drule}$ minimizes the total disutility.
\end{enumerate}
The first two notions of utilities are the natural extensions of the notions considered in \cite{talmon2019framework}. The former reflects that the voter is happy as long as the cost allocated to the project is acceptable to her, whereas the latter reflects that as the project gets more money, the voter gets happier if the cost is acceptable to her.

The second utility notion clearly assumes that if the project gets more money than what the voter thinks it deserves, voter derives zero utility. However, this need not always be the case in real-world. For example, say a voter will be happy to have an entertainment park in the neighborhood but feels that the park deserves any amount between 1000 and 5000 units. Now, if the park project is allocated 7000 units, the voter could still be happy that there is a park in the neighborhood but derive no more utility than 5000 (since it is the maximum value she thinks the park deserves). The third utility notion tackles such scenarios by modifying the second utility notion to cap the utility at \upperb{}{} instead of dropping it to $0$.

The first three utility notions assume that allocating any cost outside the range reported by the voter yields the same utility to her. However, in many situations, closer the cost of the project is to the acceptable range of the voter, higher is the satisfaction voter derives from it. For example, say a voter feels that a certain project is worth at least 1000 units. An outcome that allocates 900 units to it is likely to be more preferred by the voter over something that allocates 50 units to it. The fourth rule handles such situations by considering the distance between the cost allocated and the closest acceptable cost as the disutility. Farther the cost is from the acceptable range, higher is the amount misspent.

\section{Computational Complexity}\label{sec: computational}
Here, we analyze the computational complexity of our PB rules. We strengthen the existing positive results in the literature \cite{talmon2019framework} and also present a few new results on fixed parameter tractability. All the exact algorithms we present are based on dynamic programming. The approximation schemes we present depend on both dynamic programming as well as a clever rounding scheme.

\subsection{The Rule ${\nrule}$}\label{sec: nrule}
Recall that \nrule outputs a valid set that maximizes the sum of the utilities of voters, where the utility of a voter is defined as $\utof{}{} = \cardof{\pa \in \proj : \lowerb{}{} \leq \chocof{}{} \leq \upperb{}{}}$. Talmon and Faliszewski \cite{talmon2019framework} prove that for the case with only one permissible cost for each project, a subset maximizing the utility can be computed in polynomial time. We strengthen this result and prove that even when there are multiple permissible costs for each project, a subset maximizing the total utility can be computed in polynomial time.
\begin{theorem}\label{the: nrule-p}
    For any instance \instance, a subset $S \in \valid$ that is selected under \nrule can be computed in polynomial time.
\end{theorem}
\begin{proof}
    We present an algorithm that uses dynamic programming. Let $\scoreof{}{}$ denote the number of voters $i$ such that $\lowerb{}{} \leq \cof{}{} \leq \upperb{}{}$. Construct a dynamic programming table such that $A(x,y)$ corresponds to the cost of cheapest valid subset of $\bigcup_{j = 1}^x{\dsetof{}}$ such that the total score of projects in the set is exactly $y$.

    Let $F(y) \subseteq \dsetof{1}$ such that $\scoreof{1}{} = y$ for every $\pdof{1}{} \in \dsetof{1}$. We compute the first row as: $A(1,y) = \min(\cof{}{}: \pdof{}{} \in F(y))$. All the remaining rows are computed recursively as follows: $A(x,y) = \min\Big(A\big(x-1,y\big)\;,\;\min_{t = 1}^{\mdof{x}}{\big(A\big(x-1,y-\scoreof{x}{t}\big)+\cof{x}{t}\big)}\Big)$.
    In addition to computing this table, we store the sets corresponding to each entry in $A$ in a separate table $B$ as follows: if $A(x,y)$ achieves the minimum value at $A(x-1,y)$, we copy $B(x-1,y)$ to $B(x,y)$ and append $\pdof{}{0}$. If $A(x,y)$ achieves the minimum value at $A(x-1,y-\scoreof{x}{t})+\cof{x}{t}$ for some $t \in [1,\mdof{x}]$, we set $B(x,y) = B(x-1,y) \cup \curly{\pdof{x}{t}}$. Finally, we output the set at $B(x,y)$ such that $A(x,y)\leq \bud$ and $y$ is maximized.
    
     \paragraph{Correctness} Recurrence ensures that for any entry in $B$, we select at most one project from \dsetof{} for every \pa. While selecting the outcome from 
    $B$, we ensured that its cost is within the budget. Hence, the output of the algorithm will be a valid subset. Optimality follows from the way $A$ is defined.
    
    \paragraph{Running Time} Each row in $A$ corresponds to making a decision about one project from \proj, and hence the number of rows is $m$. Each column corresponds to a possible total score. Since we select only one project from each \dsetof{} and maximum score of any degree of a project is $n$ by definition, maximum total score is $mn$. Thus, we have $mn$ columns. Computing each entry in row $x$ takes $O(\mdof{x})$ time. Thus, the running time is $\boldsymbol{O(m^2n\tstar)}$, where $\tstar = \max_{\pa \in \proj}{\mdof{}}$.
\end{proof}

\subsection{The Rule ${\crule}$}\label{sec: crule}
Recall that \crule outputs a valid set that maximizes the sum of the utilities of voters, where the utility of a voter is: $$\utof{}{} = \suml{\pa: \lowerb{}{} \leq \chocof{}{} \leq \upperb{}{}}{\chocof{}{}}.$$ Talmon and Faliszewski \cite{talmon2019framework} prove that for the case where each project has only one permissible cost, it is \NPH to determine if there exists a feasible subset that guarantees a total utility of at least a given value. Since their model can be modeled as a special case of ours, the hardness easily follows.
\begin{proposition}\label{the: crule-nph}
    For an instance \instance and a value $s$, it is \NPH to check if \crule outputs a set with at least a total utility $s$.
\end{proposition}
\begin{proof}
Given a PB instance with approval votes (each project $j$ has some cost $c(j)$ and every voter reports a subset $A_i$ of projects that she likes), we construct an instance for \crule as follows: For every project $\pa \in \proj$, create exactly two degrees such that $\cof{}{0} = 0$ and $\cof{}{1} = c(\pa)$. For each voter $i$, set $\lowerb{}{} = 0$ for every project \pa and set $\upperb{}{} = \cof{}{1}$ if and only if $\pa \in A_i$. Clearly, both the instances are equivalent and we skip the proof of correctness.
\end{proof}
To cope up with the intractability, Talmon and Faliszewski \cite{talmon2019framework} prove that if every project has only one permissible cost, the problem has a pseudo-polynomial time algorithm and FPTAS. We extend both these results to our model with multiple permissible costs.
\begin{proposition}\label{the: crule-pseudo}
    For any instance \instance, a subset $S \in \valid$ selected under \crule can be computed in pseudo-polynomial time.
\end{proposition}
\begin{proof}
    We multiply every $\operatorname{s}(P_j^t)$ in the proof of Thm. \ref{the: nrule-p} with \cof{}{}. DP tables are also constructed as explained in Thm. \ref{the: nrule-p}. The maximum score achievable by each \pdof{}{} is $n\cof{}{}$. The total score is upper bounded by $n\!\csetof{S}$ since at most a single project from each \degreeset{} is chosen into each $B(x,y)$. This is bounded by $n\bud$ since $S \in \valid$. Thus, the table size is $O(mn\bud)$ and the time is $\boldsymbol{O(mn\tstar\bud)}$, where $\tstar = \max_{\pa \in \proj}{\mdof{}}$.
\end{proof}
\begin{theorem}\label{the: crule-fptas}
    There is an FPTAS for \crule.
\end{theorem}
\begin{proof}
    The idea is inspired from one of the existing FPTAS algorithms of the knapsack problem \cite{ibarra1975fast}. We round the scores of all the projects and use the DP table explained in Thm. \ref{the: nrule-p} on the modified instance. 

    Given an instance \instance, let $M$ be the maximum score of a degree of project, i.e., $M = \max_{\pa \in \proj, t \in [1,\mdof{}]}{\scoreof{}{}}$. We can easily ensure that $M \leq OPT$ by eliminating all the project degrees that cannot be part of any set with cost within \bud. Take any $\epsilon \in (0,1)$. Now, create new scores of all the projects as follows: $\msof{} = \floor*{\frac{s(P^t_j)m}{\epsilon M}}$. 
    We construct the DP tables similar to those in Thm. \ref{the: nrule-p}, considering new scores. 
    By the definition of \msof{}, for every $\pa \in \proj$ and $t \in [1,\mdof{}]$:
    \begin{align}
        \label{eq: nsltos}
        \msof{} &\leq \frac{\scoreof{}{}m}{\epsilon M}\\
        \label{eq: nsgtos}
        \scoreof{}{} &\leq \frac{\epsilon M (\msof{}+1)}{m}
    \end{align}
     Let $S$ be the outcome of this DP algorithm. Say $O$ denotes the optimal solution for \instance under \crule. Please note that the DP ensures that $S$ is a valid set. Since $S$ is the optimal solution for the modified scores,
    \begin{align}
        \label{eq: sgto}
        \suml{j=1}^m{\msof{\chodof{}{O}}} \leq \suml{j=1}^m{\msof{\chodof{}{S}}}
    \end{align}
    Now, we prove that the approximation factor of $(1-\epsilon)$ holds.
    \begin{align}
        \tag{From (\ref{eq: nsgtos})}
        OPT = \suml{j=1}^m{\scoreof{}{\chodof{}{O}}} &\leq \suml{j=1}^m{\frac{\epsilon M (\msof{\chodof{}{O}}+1)}{m}}\\
        \tag{From (\ref{eq: sgto})}
        &\leq \suml{j=1}^m{\frac{\epsilon M \msof{\chodof{}{S}}}{m}}+\epsilon M\\
        \tag{From (\ref{eq: nsltos})}
        &\leq \frac{\epsilon M}{m}\suml{j=1}^m{\frac{\scoreof{}{\chodof{}{S}}m}{\epsilon M}}+\epsilon M\\
        \tag{$\because M \leq OPT$}
        \suml{j=1}^m{\scoreof{}{\chodof{}{S}}} &\geq (1-\epsilon)OPT
    \end{align}
    \textit{Running Time:} The table has $m$ rows. The modified score of each degree of project is upper bounded by $\frac{s(P^t_j)m}{\epsilon M}$, which is further bounded by $\frac{m}{\epsilon}$ due to the definition of $M$. Since the output is a valid set, it has at most $m$ projects and the maximum possible total score is upper bounded by $\frac{m^2}{\epsilon}$, which is the number of columns. Computing each entry in row $x$ takes $O(\mdof{x})$ time. Thus, the running time is $\boldsymbol{O(\frac{m^3\tstar}{\epsilon})}$, where $\tstar = \max_{\pa \in \proj}{\mdof{}}$.
\end{proof}
\subsubsection{Fixed Parameter Tractability}\label{sec: crule-fpt}
Recently, Sreedurga et al. \cite{sreedurga2022maxmin} introduced a new parameter called scalable limit for PB. The authors observed that this value is often small in PB elections in real life (e.g., datasets at https://pbstanford.org/ where budget and all the costs are multiples of some really high value). Motivated by this, we prove that \crule is in FPT w.r.t. scalable limit. Before that, we define the scalable limit in our model as follows:
\begin{definition}\label{def: slimit}
    For any instance \instance, we refer to $\frac{\max_{\pdof{}{} \in \degreeproj}{\cof{}{}}}{\tiny{GCD}(\cof{1}{1},\ldots,\cof{m}{\mdof{m}},\bud)}$ as the \textbf{scalable limit} (\slimit) of the instance.
\end{definition}
Intuitively, we scale down all the costs and budget as much as possible, while ensuring that all these values continue to be integers. The scalable limit then refers to the cost of the costliest degree in $\degreeproj$ in this scaled down instance.
\begin{proposition}\label{the: crule-sl}
    For any instance \instance, computing a subset $S \in \valid$ that is selected under \crule is in \FPT w.r.t. scalable limit \slimit.
\end{proposition}
\begin{proof}
    We again prove this by constructing a DP table. Let $k =\frac{1}{\tiny{GCD}(\cof{1}{1},\ldots,\cof{m}{\mdof{m}},\bud)}$. First, we get a new instance $\instance'$ by scaling down all the costs and budget as follows: $\cof{}{'t} =k\cof{}{}$ and $\bud' = k\bud$. Then, construct the DP table similar to \Cref{the: nrule-p}. 

    \textit{Running Time:} The table has $m$ rows. The modified cost of each degree of project is upper bounded by \slimit and hence the maximum score possible from each degree of a project is upperbounded by $n\slimit$. Since the output is a valid set, there can be at most $m$ projects in it and hence the maximum possible total score is upper bounded by $mn\slimit$, which is the number of columns. Computing each entry in row $x$ takes $O(\mdof{x})$ time and the running time is $\boldsymbol{O(m^2n\slimit \tstar)}$.
\end{proof}
\begin{proposition}
    For any instance \instance, computing a subset $S \in \valid$ that is selected under \crule is in \FPT w.r.t. $m$.
\end{proposition}
The above proposition follows from the fact the total number of valid subsets is exponential in $m$ (upper bounded by $(\tstar)^m$, where $\tstar = \max_{\pa \in \proj}{\mdof{}}$) and computing the total utility of any set under \crule can be done in polynomial time.

\subsection{The Rule ${\ccaprule}$}\label{sec: ccaprule}
This rule is similar to the previous rule \crule, with the only difference being the way a voter is assumed to view the projects allocated higher cost than the approved limit. The utility is defined as $\utof{}{} = \sum_{\pa \in \proj}{\utof{}{j}}$, where \utof{}{j} is taken as defined in \Cref{sec: utilities}. The proof of the following result is deferred to the appendix.

\begin{proposition}\label{the: ccaprule}
    Given an instance \instance, the following statements hold:
    \begin{enumerate}
        \item For any value $s$, it is \NPH to determine if \ccaprule outputs a set that has a total utility of at least $s$.
        \item A subset $S \in \valid$ that is selected under \ccaprule can be computed in pseudo-polynomial time.
        \item There is an FPTAS for \ccaprule.
        \item Computing a subset $S \in \valid$ that is selected under \ccaprule is in \FPT w.r.t. scalable 
 limit \slimit.
        \item Computing a subset $S \in \valid$ that is selected under \crule is in \FPT w.r.t. $m$.
    \end{enumerate}
\end{proposition}

\subsection{The Rule ${\drule}$}\label{sec: drule}
Recall that \drule outputs the feasible set that minimizes $\dutof{}{} = \sum_{\pa \in \proj}{\dutof{}{j}}$, where \dutof{}{j} is defined as in \Cref{sec: utilities}. It is worth bearing in mind that the results from \Cref{sec: crule} and \Cref{sec: ccaprule} cannot be transferred since we have a minimization problem and the fact that we cannot upper bound the disutilities by \bud. We start the section by proving that this rule too is \NPH, using a slightly different but simple reduction from the same problem (used in \Cref{the: crule-nph}) solved by Talmon and Faliszewski \cite{talmon2019framework}. Then, we give a parameterized approximation algorithm, more specifically a parameterized FPTAS, which guarantees the approximation ratio when the parameter is small \cite{feldmann2020survey}.
\begin{proposition}\label{the: drule-nph}
    For an instance \instance and a value $s$, it is \NPH to determine if \drule outputs a set that has a total disutility of at most $s$.
\end{proposition}
\begin{proof}
For a PB instance with approval votes (each project $j$ has some cost $c(j)$ and every voter reports a subset $A_i$ of projects that she likes) and a positive value $x$, we construct an instance for \crule as follows: For every project $\pa \in \proj$, create exactly two degrees such that $\cof{}{0} = 0$ and $\cof{}{1} = c(\pa)$. For each voter $i$, set $\upperb{}{} = c(\pa)$ for every project \pa and set $\lowerb{}{} = \cof{}{1}$ if and only if $\pa \in A_i$. Let $Z = \sum_{i \in \voters}{\csetof{A_i}}$. Set $s = Z-x$.

\paragraph{Correctness} To prove the forward direction, say the approval-based PB instance is a \yes instance. That implies, there exists a set $S$ of projects such that $\sum_{i \in \voters}{\csetof{S \cap A_i}} \geq x$. Now calculate the disutility of $S$ under \drule. Then, $\dutof{}{j} = c(\pa)$ for every $j \in A_i \setminus S$ and $0$ otherwise. Therefore, $\dutof{}{} = -c(A_i)-c(A_i \cap S)$. Then, total utility from $S$ = $Z-\sum_{i \in \voters}{c(A_i \cap S)}$. Since the former is at least x, this total disutility is at most Z-x. The backward direction can be argued similarly.
\end{proof}
\subsubsection{Parameterized Approximation Algorithm (FPTAS)}\label{sec: drule-pfptas}
The very high disutilities that are hard to be bound motivate us to propose a \emph{parameterized FPTAS}, which has been of a great interest recently \cite{feldmann2020survey}. We consider a parameter that we call \emph{\variance}, \vpara, which, on being given an instance \instance, intuitively shows how divergent the disutilities of different degrees of projects in the instance are. It is captured by measuring the highest disutility a single project can have, relative to sum of the least possible disutilities from all the projects. We explain this formally below.

Given an instance \instance, for each degree of a project $\pdof{}{} \in \degreeproj$, we use $\qof{}{}$ to denote the disutility the project will contribute to any set that contains it. We call this \emph{disutility contribution} of \pdof{}{}. Suppose for example, $\pdof{}{} \in S$. That implies, $\chocof{}{} = \cof{}{}$ ($\because S$ is valid). Therefore, $$\qof{}{} = \sum_{i: \cof{}{}<\lowerb{}{}}{(\lowerb{}{}-\cof{}{})}+\sum_{i: \upperb{}{}<\cof{}{}}{(\cof{}{}-\upperb{}{})}.$$ It can be observed that, $\sum_{i \in \voters}\dutof{}{} = \sum_{\pdof{}{} \in S}{\qof{}{}}$.

Let $\qmax = \max_{\pdof{}{} \in \degreeproj}{\qof{}{}}$. That is, $\qmax$ is the maximum of disutility contributions of all the degrees of all the projects. We use $\qsum$ to denote the sum of minimum disutility contributions from each \dsetof{}, i.e., $\qsum = \sum_{\pa \in \proj}{\min_{t \in \dsetof{}}{\qof{}{}}}$. Our parameter \vpara is the ratio $\qmax/\qsum$. Intuitively, this parameter means that we want the disutility contribution from a single project not to be extremely higher than the sum of the least disutility contributions from all the $m$ projects.
\begin{theorem}
    For any instance \instance and $\epsilon\in [0,1]$, a subset $S \in \valid$ such that $\sum_{i \in \voters}{\dutof{}{}} \leq (1+\epsilon)OPT$ can be computed in $O(\frac{m^3\vpara\tstar}{\epsilon})$ time, where $OPT$ is the optimal possible total disutility under \drule, \vpara is the \variance, and $\tstar = \max_{\pa \in \proj}{\mdof{}}$. 
\end{theorem}
\begin{proof}
    The idea is similar to that in Thm. \ref{the: crule-fptas}. We define new disutilities for each project by rounding their disutility contributions. We then construct a DP table using which we select a set. For each project \pdof{}{}, we define \msof{} as follows: $$\msof{} = \floor*{{\qof{}{}m}/{\epsilon\qsum}}.$$
    
    We construct the DP tables $A$ and $B$ similar to those in Thm. \ref{the: nrule-p}, with a slight change the each column now represents the total disutility, and \scoreof{}{} is replaced by \msof{} defined in the above para. 
    We output the set at $B(x,y)$ such that $A(x,y) \leq \bud$ and $y$ is \emph{minimized}. Let $S$ be the resultant outcome.
    Say $O$ denotes the optimal solution for \instance under \drule. Please note that the DP ensures that $S$ is a valid set. Since $S$ is the optimal solution w.r.t. new disutilities (i.e., w.r.t. $\alpha$'s),
    \begin{align}
        \label{eq: sgto'}
        \suml{j=1}^m{\msof{\chodof{}{S}}} \leq \suml{j=1}^m{\msof{\chodof{}{O}}}
    \end{align}
    By the definition of $\alpha$, we have
    \begin{align}
        \label{eq: ndltod}
        \qof{}{} &\geq \frac{\epsilon\qsum\msof{}{}}{m}\\
        \label{eq: ndgtod}
        \msof{}{} &\geq \frac{\qof{}{}m}{\epsilon\qsum}-1
    \end{align}
    Now, we prove that the approximation factor of $(1+\epsilon)$ holds.
    \begin{align}
        \tag{From (\ref{eq: ndltod})}
        OPT = \suml{j=1}^m{\qof{}{\chodof{}{O}}} &\geq \suml{j=1}^m{\frac{\epsilon \qsum \msof{\chodof{}{O}}}{m}}\\
        \tag{From (\ref{eq: sgto'})}
        &\geq \suml{j=1}^m{\frac{\epsilon \qsum \msof{\chodof{}{S}}}{m}}\\
        \tag{From (\ref{eq: ndgtod})}
        &\geq \left(\frac{\epsilon \qsum}{m}\suml{j=1}^m{\frac{m\qof{}{\chodof{}{S}}}{\epsilon \qsum}}\right)-\epsilon \qsum\\
        \label{eq: lastbutone}
        \suml{j=1}^m{\qof{}{\chodof{}{S}}} &\leq OPT+\epsilon \qsum
    \end{align}
\qsum is the sum of minimum possible disutility contribution from each \dsetof{}. Any valid set should have one degree (degree $0$ corresponds to not funding the project) from each \dsetof{}. Since the optimal solution has to be a valid set, the optimal disutility must be at least \qsum, i.e., $OPT \geq \qsum$. Applying this in Eqn. \ref{eq: lastbutone}, results in $\sum_{j=1}^m{\qof{}{\chodof{}{S}}} \leq (1+\epsilon)OPT$.

\paragraph{Running Time} Now, we need to check the running time of constructing the DP table. The table has $m$ rows. The disutilities calculated for each degree of project is upper bounded by $\frac{\qof{}{}m}{\epsilon\qsum}$, which is further bounded by $\frac{m\vpara}{\epsilon}$ (by the definition of $\vpara$, $\because \qof{}{}\leq \qmax$). Since the output is a valid set, there can be at most $m$ projects in the output and hence the maximum possible total disutility is upper bounded by $\frac{m^2\vpara}{\epsilon}$, which is the number of columns. Computing each entry in row $x$ takes $O(\mdof{x})$ time. Thus, the running time is $\boldsymbol{O(\frac{m^3\vpara\tstar}{\epsilon})}$, where $\tstar = \max_{\pa \in \proj}{\mdof{}}$.
\end{proof}
\paragraph{Fixed Parameter Tractability} Here, we prove that \drule is also in FPT with respect to the parameters \slimit or $m$.
\begin{proposition}\label{the: drule-sl}
    For any instance \instance, computing a subset $S \in \valid$ that is selected under \drule is in \FPT w.r.t. scalable limit \slimit.
\end{proposition}
The proof of the above proposition is similar to the proof of \Cref{the: crule-sl}. We scale down the costs of the projects and the budget as explained in that proof. Followed by this, the DP table and the output will be computed as explained in \Cref{the: drule-sl}. The running time analysis is exactly the same as explained in \Cref{the: crule-sl}

\begin{proposition}
    For any instance \instance, computing a subset $S \in \valid$ that is selected under \drule is in \FPT w.r.t. $m$.
\end{proposition}
The above result follows from the fact that when $m$ is constant, we can exhaustively search all the valid sets for the set with optimal disutility.

\section{Budgeting Axioms}\label{sec: axioms}
An enormous amount of work has been done on the axiomatic study of PB with approval votes \cite{aziz2018proportionally,aziz2019proportionally,talmon2019framework,rey2020designing,baumeister2020irresolute,sreedurga2022maxmin}. However, our PB model in which each project has a set of permissible costs is unique technically as well as realistically. This uniqueness demands the development of novel axioms applicable explicitly in such a setting. We introduce different axioms and investigate which of the PB rules introduced in \Cref{sec: utilities} satisfy our axioms and which do not. We defer some proofs in this section to appendix.

The first axiom implies that if the voters are narrowing down their interval towards a degree that was winning, then the degree still continues to win.
\begin{definition}[Shrink-resistant]\label{def: shrink}
	A PB rule \pbrule is said to be \emph{shrink-resistant} if for any instance \instance, a voter \ii, any $\pa \in \proj$, it holds that a set $S$ selected under \pbrule continues to be selected even if \lowerb{}{} and \upperb{}{} are shifted closer to \chocof{}{}.
\end{definition}
\begin{proposition}\label{the: shrink}
	All the four rules \nrule, \crule, \ccaprule, and \drule are shrink-resistant.
\end{proposition}
Given an instance \instance, for each project $\pa\in \proj$, we define $\conrange = \curly{\cof{}{}: \pdof{}{} \in \degreeproj,\;\forall i \in \voters\;\; \lowerb{}{} \leq \cof{}{} \leq \upperb{}{}}$. Let $\maxconrange = \max{(\conrange)}$ and $\minconrange = \min{(\conrange)}$. The second axiom requires that if some range of costs is found to be acceptable unanimously by all the voters, then the cost allocated to the project must not go beyond this range.
\begin{definition}[Range-abiding]\label{def: rangeabiding}
	A PB rule \pbrule is said to be \emph{range-abiding} if for any instance \instance, a project $\pa \in \proj$, and a set $S$ selected under \pbrule, it holds that $$\conrange \neq \emptyset \implies \chocof{}{} \leq \maxconrange.$$
\end{definition}
\begin{proposition}\label{the: range-abiding}
	The rules \nrule and \drule are range-abiding, whereas \crule and \ccaprule are not.
\end{proposition}
\begin{proof}
	First, we prove that \crule and \ccaprule are not range-abiding. Recall that \bud denotes the budget. Say we have a single project \pa with exactly two permissible costs, respectively equal to $\floor*{\frac{\bud-1}{n}}$ and \bud. Say all the voters set $\lowerb{}{} = \cof{}{1}$. One voter \ii sets $\upperb{}{} = \cof{}{2}$, whereas all the remaining voters report $\cof{}{1}$ as the upper bound. Clearly, in a set $S$ containing \pdof{}{1}, the utility from project \pa is at most $\bud-1$. Whereas, if $S$ selects \pdof{}{2}, the utility from \pa is equal to $\bud$ (but all due to the utility of a single voter). Thus $\pdof{}{2}$ is selected. Note that, in this example, $\conrange = \floor*{\frac{\bud-1}{n}}$ and clearly $\cof{}{2} > \conrange$.

    Now, assume that \nrule and \drule are not range-abiding. That is, in a set $S$ selected under the rule, $\chocof{}{} > \maxconrange$. Consider the set $S' = S\setminus\curly{\chodof{}{}} \cup \curly{\pdof{}{}:\maxconrange = \cof{}{}}$. Clearly, $S'$ is valid since $\chocof{}{} > \maxconrange$ and $S$ is valid. Also, $\utofd{}{\pa} > \utof{}{\pa}$ since the degree chosen in $S'$ is within the bounds reported by every voter (it gives the optimal utility of $n$ for \nrule and the optimal disutility of $0$ for \drule). Thus, $S'$ is a strictly better set and $S$ must not be selected. This forms a contradiction.
\end{proof}
The next axiom requires that increasing the budget should result in the winning degree of some project moving closer to the range of costs for it found to be acceptable unanimously by all the voters.
\begin{definition}[Range-converging]\label{def: rangeconverging}
    A PB rule \pbrule is said to be \emph{range-converging} if for any instance \instance, a set $S$ selected for \instance under \pbrule, and a set $S' \neq S$ selected under \pbrule on increasing the budget, it holds that whenever there is at least one project $k$ with $\tau_k \neq \emptyset$, there also exists some project $\pa \in \proj$ such that: $$\chocof{}{} \notin \conrange \implies |\chocof{}{} - \minconrange| > |\chocof{}{S'} - \minconrange|.$$
\end{definition}
\begin{proposition}\label{the: range-converging}
    All the four rules \nrule, \crule, \ccaprule, and \drule are range-converging.
\end{proposition}
The next axiom requires that if some range of costs is found to be acceptable unanimously by all the voters, the project must be allocated the maximum amount in this range.
\begin{definition}[Range-unanimous]\label{def: rangeunanimous}
    A PB rule \pbrule is said to be \emph{range-unanimous} if for any instance \instance, whenever $\sum_{\pa \in \proj}{\maxconrange} \leq \bud$, the set $\curly{\pdof{}{}: \pa \in \proj,\;\cof{}{} = \maxconrange}$ is selected under \pbrule.
\end{definition}
Note that the above holds by default if \maxconrange is not defined for some \pa. Range-unanimity and range-abidingness do not imply each other, though they seem closely related. For example, take the rule that picks for each project \pa, degree with minimum cost in \conrange (whenever this set costs lesser than \bud). This is range-abiding but not range-unanimous. In an instance where the maximum costs in \conrange together cost more than the budget, range-unanimity is satisfied by default by any rule. However, a rule that selects a degree for project \pa whose cost is greater than maximum cost in \conrange, does not satisfy range-abidingness. Thus, range-unanimity does not imply range-abidingness.
\begin{proposition}\label{the: range-unanimous}
    The rules \nrule and \drule are range-unanimous, whereas \crule and \ccaprule are not.
\end{proposition}

The following axiom essentially implies that if two valid sets differ only on the degree of one project, then the set with higher degree needs to be preferred.

\begin{definition}[Degree-efficient]\label{def: degreeefficient}
    A PB rule \pbrule is said to be \emph{degree-efficient} if for any instance \instance, any project $\pa \in \proj$, any set $S$ selected under \pbrule, and the degree $x \in \chodof{}{}$, it holds that $$k>x\implies\csetof{}-\chocof{}{}+\cof{}{k} > \bud.$$
\end{definition}
\begin{proposition}\label{the: degreeefficient}
    The \ccaprule is degree-efficient, whereas \nrule, \crule, and \drule are not.
\end{proposition}
\begin{proof}
    Let \pbrule be \nrule, \crule, or \drule. Consider the example with single project $\pa$ such that for every $i$: (i) $\cof{}{\mdof{}} \!\!<\!\! \bud$ and $\upperb{}{} \!\!<\!\! \cof{}{\mdof{}}$ (ii) there exists $t \!\!<\!\! \mdof{}$ such that $\lowerb{}{} \!\!\leq\!\! \cof{}{} \!\!\leq\!\! \upperb{}{}$. The set $\curly{\pdof{}{}: \cof{}{} = \maxconrange}$ is selected under \pbrule since we know that $\conrange \neq \emptyset$ and $\maxconrange < \bud$. Assume that the axiom is satisfied. Then, by the definiton of \mdof{}, $\cof{}{\mdof{}}>\bud$ contradicting (i). The proof for \ccaprule is deferred to appendix.
\end{proof}

The next two axioms insist that the valid set closer to the bounds reported by all the voters must be preferred over a valid set farther from them.
\begin{definition}[Lower bound-sensitive]\label{def: lboundsensitive}
    A PB rule \pbrule is said to be \emph{lower bound-sensitive} if for any instance \instance, any project $\pa \in \proj$, and any two valid set $S,S'$ such that for every voter $i$ we have $\chocof{}{S}\!\!<\!\!\chocof{}{S'}\!\!<\!\!\lowerb{}{}$, $S$ is not selected under \pbrule.
\end{definition}
\begin{proposition}\label{the: lboundsensitive}
    The \drule is lower bound-sensitive, whereas \nrule, \crule, and \ccaprule are not.
\end{proposition}
\begin{proof}
    The proof for \drule is deferred to appendix. Let \pbrule be \nrule, \crule, or \ccaprule. For lower-bound sensitivity, consider a counter example as follows: (i) there are two projects with $\dsetof{1} = \curly{\pdof{1}{0},\pdof{1}{1},\pdof{1}{2},\pdof{1}{3}}$ and $\dsetof{2} = \curly{\pdof{2}{0},\pdof{2}{1}}$ (ii) $\cof{1}{1} = 1$, $\cof{1}{2} = 2$, $\cof{1}{3} = \bud-3$, and $\cof{2}{1} = \bud-2$ (iii) for every voter $i\in\voters$, $\lowerb{}{1} = \upperb{}{1} = \cof{1}{3}$ and $\lowerb{}{2} = \upperb{}{2} = \cof{2}{1}$. Clearly, $S= \curly{\pdof{1}{1},\pdof{2}{1}}$ is a set that is selected under \pbrule. Set $S' = \curly{\pdof{1}{2},\pdof{2}{1}}$ and $j = 1$.
\end{proof}
\begin{definition}[Upper bound-sensitive]\label{def: uboundsensitive}
    A PB rule \pbrule is said to be \emph{upper bound-sensitive} if for any instance \instance, any project $\pa \in \proj$, and any two valid set $S,S'$ such that for every voter $i$ we have $\chocof{}{S}\!\!>\!\!\chocof{}{S'}\!\!>\!\!\upperb{}{}$, $S$ is not selected under \pbrule.
\end{definition}
\begin{proposition}\label{the: uboundsensitive}
    The rules \nrule,\crule, and \drule are upper bound-sensitive, whereas \ccaprule is not.
\end{proposition}
Next, we extend discount montonicity \cite{talmon2019framework} to ensure that a winning degree of project shouldn't be omitted if it becomes less expensive.
\begin{definition}[Discount-proof]\label{def: discount}
    A PB rule \pbrule is said to be \emph{discount-proof} if for any instance \instance, any project $\pa \in \proj$, and a set $S$ that is selected under \pbrule, $S$ continues to be selected if \chocof{}{} is decreased by $1$.
\end{definition}
\begin{proposition}\label{the: discount}
    The \nrule is discount-proof, whereas \crule, \ccaprule, and \drule are not.
\end{proposition}

\section{Summary}\label{sec: conclusion}
\vspace*{-0.3\baselineskip}
Many times, there are multiple ways of executing a public project and hence several, but limited number of, choices for the amount to be allocated to this project. Unfortunately, the existing preference elicitation methods and aggregation rules for participatory budgeting do not take this factor into account. Our work is an attempt to bridge this gap. We generalized two utility notions defined for PB with approval votes to our model. We also proposed two other utility notions unique to our model. We analyzed all the corresponding utilitarian PB rules computationally and axiomatically.

Our computational part strengthens all the existing positive results, and also introduces several new parameterized tractability results (FPT, parameterized FPTAS) taking into account the parameters recently introduced in the PB literature. It is worth highlighting that all our computational results in \Cref{sec: computational} can be generalized by replacing the utilities with cardinal utility for every degree of each project. However, we present all the results for ranged approval votes due to their practical relevance, simplicity, and deep association with axiomatic analysis.

Followed by this, we introduce several axioms for our model with ranged approval votes and investigate which of these are satisfied by our PB rules. Note that, though none of the proposed PB rules satisfies all the axioms, every rule satisfies some axioms. Axiomatic analysis reflects the properties of each rule, using which the PB organizer can pick a rule based on the context. Also, it is worth bearing in mind that the novel disutility notion and PB rule \drule we proposed for our model satisfies as many axioms as any simple approval-based PB rule satisfies. One of the key takeaways of this paper is hence a conclusion that \drule is a very good choice when each voter approves a range of costs.
\section*{Acknowledgements}
Sreedurga gratefully acknowledges the support of prime minister research fellowship (PMRF) by Government of India, and also thanks Prof. Ulle Endriss, Simon Rey, and Dr. Jan Maly for their helpful feedback.



\bibliographystyle{plain}
\small{\bibliography{paper.bib}}
\pagebreak
\appendix
\section{Appendix}
\newtheorem*{ccaprule_rep}{\Cref{the: ccaprule}}
\begin{ccaprule_rep}
Given an instance \instance, the following statements hold:
    \begin{enumerate}
        \item For any value $s$, it is \NPH to determine if \ccaprule outputs a set that has a total utility of at least $s$.
        \item A subset $S \in \valid$ that is selected under \ccaprule can be computed in pseudo-polynomial time.
        \item There is an FPTAS for \ccaprule.
        \item Computing a subset $S \in \valid$ that is selected under \ccaprule is in \FPT w.r.t. scalable 
 limit \slimit.
        \item Computing a subset $S \in \valid$ that is selected under \crule is in \FPT w.r.t. $m$.
    \end{enumerate}
\end{ccaprule_rep}
\begin{proof}
    All the above statements follow from the proofs in \Cref{sec: crule} with some minor changes. Statements $(1)$ and $(5)$ follow without any changes in the proof (\ccaprule is the same as \crule when every project has only permissible cost).

    For (2), we modify the definition of \scoreof{}{} as $\scoreof{}{} = \suml{i \in \voters}{\utof{}{j}}$, where $\utof{}{j}$ is as defined in the rule. The maximum possible total score of $S$ is again $n\csetof{S}$ (since $\utof{}{j} \leq \chocof{}{S}$ always) and the rest follows. For (3), we use the above $\scoreof{}{}$ in the proof of \Cref{the: crule-fptas} and the rest is same. The proof for (4) follows from \Cref{the: crule-sl} since $\utof{}{j}$ is upper-bounded by \cof{}{} for every $i \in \voters$ and $\pa \in \proj$.
\end{proof}

\newtheorem*{shrink_rep}{\Cref{the: shrink}}
\begin{shrink_rep}
All the four rules \nrule, \crule, \ccaprule, and \drule are shrink-resistant.
\end{shrink_rep}
\begin{proof} 
 	First, we look at \nrule. Consider a set $S$ selected under \nrule. If initially $\lowerb{}{} \leq \chocof{}{} \leq \upperb{}{}$, then the condition continues to satisfy even if the \lowerb{}{} and \upperb{}{} are shifted closer to \chocof{}{}. Thus, \utof{}{} does not change. If initially $\chocof{}{} < \lowerb{}{}$ or $\upperb{}{} < \chocof{}{}$, then shifting \lowerb{}{} and \upperb{}{} closer to \chocof{}{} may cause \utof{}{} remain the same or increase by $1$. Utilities of sets without \chodof{}{} remain unchanged thus proving the claim. This logic can also be extended to \crule and \ccaprule. For both these rules, shifting will make \utof{}{} remain the same or increased by \chocof{}{}. Finally, look at \drule. By the definition of \dutof{}{}, shifting \lowerb{}{} and \upperb{}{} closer to \chocof{}{} will cause a strict decrease in \dutof{}{} if it was a non-zero value. Else, \dutof{}{} remains the same.
\end{proof}

\newtheorem*{range-converging_rep}{\Cref{the: range-converging}}
\begin{range-converging_rep}
All the four rules \nrule, \crule, \ccaprule, and \drule are range-converging.
\end{range-converging_rep}
\begin{proof}
    First, consider the first three rules. Let \pbrule be \nrule, \crule, or \ccaprule. Note that for these rules, the value $\sum_{i \in \voters}{\utof{}{\pa}}$ decreases as $\chocof{}{}$ moves farther from all the costs in \conrange (the farther \chocof{}{} moves, more is the number of voters for whom \chocof{}{} falls out of their acceptable range). We prove the claim by contradiction. Assume that for any project $\pa \in \proj$ such that $\chocof{}{} \notin \conrange$ and $\chocof{}{} \neq \chocof{}{S'}$, $\chocof{}{}$ is closer to \conrange than \chocof{}{S'} is. This implies that $\sum_{i \in \voters}{\utofd{}{\pa}} < \sum_{i \in \voters}{\utof{}{\pa}}$. By adding these inequalities for all such $\pa$, we have $\sum_{i \in \voters}{\utofd{}{}} < \sum_{i \in \voters}{\utof{}{}}$. Thus, $S'$ cannot be selected by \pbrule since $S$ continues to be feasible under the increased budget.  
    The proof for rule \drule follows similar idea. This is because, $\sum_{i \in \voters}{\dutof{}{\pa}}$ increases as $\chocof{}{}$ moves farther from all the costs in \conrange (as \chocof{}{} falls out of the acceptable range for more voters).
\end{proof}

\newtheorem*{range-unanimous_rep}{\Cref{the: range-unanimous}}
\begin{range-unanimous_rep}
The rules \nrule and \drule are range-unanimous, whereas \crule and \ccaprule are not.
\end{range-unanimous_rep}
\begin{proof}
    The proof for \crule and \ccaprule follow from the example in the proof of \Cref{the: range-abiding}. Now, consider the rules \nrule and \drule. Note that the set $\curly{\pdof{}{}: \pa \in \proj,\;\cof{}{} = \maxconrange}$ achieves the optimal total utility $mn$ and optimal total disutility of $0$ respectively for both the rules. Hence, if $S'$ not being selected implies that $\csetof{S'} > \bud$. This contradicts $\suml{\pa \in \proj}{\maxconrange} \leq \bud$ and completes the argument. 
\end{proof}

\newtheorem*{degreeefficient_rep}{\Cref{the: degreeefficient}}
\begin{degreeefficient_rep}
The \ccaprule is degree-efficient, whereas \nrule, \crule, and \drule are not.
\end{degreeefficient_rep}
\begin{proof}
    We only prove for \ccaprule here since the rest is covered in the main paper. Let $S$ be a selected set and $k < \mdof{}$ be such that $k > \chodof{}{}$. Consider the set $S' = S\setminus\curly{\chodof{}{}}\cup \curly{\pdof{}{k}}$. Take any voter $i$. If $\chocof{}{} < \lowerb{}{}$, $\utofd{}{j}$ remains zero like $\utof{}{j}$ or increase by \cof{}{k}. If $\lowerb{}{}\leq\chocof{}{}\leq\upperb{}{}$, $\utofd{}{j}$ is exactly $\min{(\cof{}{k},\upperb{}{})}-\chocof{}{}$ more than $\utof{}{j}$. This value is positive since $k > \chodof{}{}$. Finally, if $\chocof{}{} > \upperb{}{}$, $\utofd{}{j}$ is exactly $\cof{}{k}$ more than $\utof{}{j}$. Thus, \utofd{}{} is clearly greater than \utof{}{}. Since $S$ is selected under \drule, $S'$ must be infeasible. Thus, the given condition holds.
\end{proof}

\newtheorem*{lboundsensitive_rep}{\Cref{the: lboundsensitive}}
\begin{lboundsensitive_rep}
The \drule is lower bound-sensitive, whereas \nrule, \crule, and \ccaprule are not.
\end{lboundsensitive_rep}
\begin{proof}
    We only prove for \drule here since the rest is covered in the main paper. The proof is straight-forward. Since the disutility of $i$ from $S$ depends on $\lowerb{}{}-\chocof{}{}$ and $\chocof{}{}-\upperb{}{}$, clearly, $\dutof{}{} < \dutofd{}{}$ for every voter $i$. And hence $S$ does not get selected by \drule.
\end{proof}

\newtheorem*{uboundsensitive_rep}{\Cref{the: uboundsensitive}}
\begin{uboundsensitive_rep}
The rules \nrule,\crule, and \drule are upper bound-sensitive, whereas \ccaprule is not.
\end{uboundsensitive_rep}
\begin{proof}
    First, let the rule \pbrule be \nrule, \crule, or \drule. Since $\chocof{}{} > \upperb{}{}$, $\utof{}{j} = 0$ and $\dutof{}{} = \chocof{}{}-\upperb{}{}$. Clearly,  utility of $S$ and $S'$ will be same for \nrule and \crule. However, none of them chooses $S$ or $S'$. For \drule, utility from $S$ will be lesser than that from $S'$. Thus, $S$ cannot be selected. Finally, consider \ccaprule. Consider an instance such that: (i) there are two projects with $\dsetof{1} = \curly{\pdof{1}{0},\pdof{1}{1},\pdof{1}{2},\pdof{1}{3}}$ and $\dsetof{2} = \curly{\pdof{2}{0},\pdof{2}{1}}$ (ii) $\cof{1}{1} = 1$, $\cof{1}{2} = 2$, $\cof{1}{3} = 3$, and $\cof{2}{1} = \bud-3$ (iii) for every voter $i\in\voters$, $\lowerb{}{1} = \upperb{}{1} = \cof{1}{1}$ and $\lowerb{}{2} = \upperb{}{2} = \cof{2}{1}$. Clearly, $S= \curly{\pdof{1}{3},\pdof{2}{1}}$ is a set that is selected under \ccaprule. Set $S' = \curly{\pdof{1}{2},\pdof{2}{1}}$ and $j = 1$. The condition is not met.
\end{proof}

\newtheorem*{discount_rep}{\Cref{the: discount}}
\begin{discount_rep}
The \nrule is discount-proof, whereas \crule, \ccaprule, and \drule are not.
\end{discount_rep}
\begin{proof}
    Let the rule \pbrule be \crule, \ccaprule, or \drule. Consider an instance such that: (i) there are two projects with $\dsetof{1} = \curly{\pdof{1}{0},\pdof{1}{1}}$ and $\dsetof{2} = \curly{\pdof{2}{0},\pdof{2}{1}}$ (ii) $\cof{1}{1} = 2$ and $\cof{2}{1} = 2$, $\bud=2$ (iii) for every voter $i\in\voters$, $\lowerb{}{1} = \upperb{}{1} = \cof{1}{1}$ and $\lowerb{}{2} = \upperb{}{2} = \cof{2}{1}$. Clearly, $S= \curly{\pdof{1}{1}}$ is a set that is selected under \pbrule. If the \cof{1}{1} is changed to $1$, $S$ is not selected since only $\curly{\pdof{2}{1}}$ is the unique set to be selected.

    Now, consider \nrule. This is straight-forward since the utilities are unaffected by the costs as long as they are in the acceptable range. Please note that by their definitions, both \lowerb{}{} and \upperb{}{} belong to the set of permissible costs and hence if the cost that is reduced is equal to the lower bound reported by some voter, the lower bound is also automatically considered to be reduced by $1$. Thus, the number of voters finding the chosen degree to be acceptable remains unaffected.
\end{proof}
\end{document}